\documentclass[12pt,leqno]{article}

\input{style}
\usepackage{enumitem}
\newcounter{item}
\numberwithin{equation}{section}

\begin{document}

\title{Stability and analytic expansions of local solutions of systems
  of quadratic BSDEs with applications to a price impact
  model\thanks{We are grateful to Kostas Kardaras and two anonymous
    reviewers for their insightful comments on earlier versions of the
    paper.}}

\author{Dmitry Kramkov\thanks{The author also holds a part-time
    position at the University of Oxford. This research was supported
    in part by the Carnegie Mellon-Portugal Program and by the
    Oxford-Man Institute for Quantitative Finance at the University of Oxford.} \\
  Carnegie Mellon University, \\ Department of Mathematical Sciences, \\
  5000 Forbes Avenue, Pittsburgh, PA, 15213-3890, US \\
  (kramkov@cmu.edu) \and Sergio Pulido\thanks{The author's research
    benefited from the support of the ``Chair Markets in Transition"
    under the aegis of Louis Bachelier laboratory, a joint initiative
    of \'Ecole Polytechnique, Universit\'e d'\'Evry-Val-d'Essonne and
    F\'ed\'eration Bancaire Fran\c caise, Labex ANR 11-LABX-0019. This research has also
    received funding from the European Research Council under the
    European Union's Seventh Framework Programme (FP/2007-2013) / ERC
    Grant Agreement n.~307465-POLYTE.}\\ Laboratoire de
  Math\'ematiques et Mod\'elisation d'\'Evry (LaMME),\\ Universit\'e
  d'\'Evry-Val-d'Essonne, ENSIIE, UMR CNRS 8071,\\
  IBGBI, 23 Boulevard de France, 91037 \'Evry Cedex, France\\
  (sergio.pulidonino@ensiie.fr)} \date{\today}

\maketitle
\pagebreak
\begin{abstract}
  We obtain stability estimates and derive analytic expansions for
  local solutions of multidimensional quadratic backward stochastic
  differential equations. We apply these results to a financial model
  where the prices of risky assets are quoted by a representative
  dealer in such a way that it is optimal to meet an exogenous
  demand. We show that the prices are stable under the demand process
  and derive their analytic expansions for small risk aversion
  coefficients of the dealer.
\end{abstract}

\begin{description}
\item[Keywords:] multidimensional quadratic BSDEs, stability of
  quadratic BSDEs, asymptotic behavior of quadratic BSDEs, liquidity,
  price impact.
\item[AMS Subject Classification (2010):] 60H10, 91B24, 91G80.
\item[JEL Classification:] D53, G12, C62.
\end{description}

\section{Introduction}
\label{sec:intro}

One-dimensional backward stochastic differential equations (BSDEs)
with quadratic growth are well studied. Existence, uniqueness, and
stability of their solutions for bounded terminal conditions have been
established in the pioneering paper \citet{Kobyl:00}. Alternative
proofs have been proposed in~\citet{Tevz:08} and~\citet{BrianElie:13}.
Generalizations to the unbounded case have been obtained
in~\citet{BrianHu:06, BrianHu:08} among others.

The situation with systems of quadratic BSDEs is more
cumbersome. Unless there is a special structure (as, for instance,
in~\citet{darling:95},~\citet{tang:03},~\citet{ElKHamad:03},~\citet{CherNam:15},
and~\citet{HuTang:16}), they may fail to have a solution even with
bounded terminal conditions. A counterexample can be found
in~\citet{FreiReis:11}.  On a positive side, existence and local
uniqueness of solutions have been obtained for sufficiently small
terminal conditions, first, in~\citet{Tevz:08} for the
$\Lspace{\infty}$-norm and then in~\citet{Frei:14}
and~\citet{KramkovPulido:16} for the bounded mean oscillation ($\bmo$)
norm.

In this paper, we establish stability properties and derive analytic
expansions of such local solutions.  Our main results are stated in
Theorems~\ref{th:1} and~\ref{th:2}. In Theorem~\ref{th:1}, we get
stability estimates in $\Sspace{p}$ and $\bmo$ spaces with respect to
the driver and the terminal condition. In Theorem~\ref{th:2} we obtain
analytic expansions in $\bmo$ with respect to the terminal
condition. The coefficients of these power series can be calculated
recursively up to an arbitrary order.

This work is motivated by our study in \cite{KramkovPulido:16} of a
price impact model from the market microstructure theory, which builds
upon earlier works by~\citet{GrosMill:88},~\citet*{GarlPederPotes:09},
and~\citet{Germ:11}. In this model, a representative dealer provides
liquidity for risky stocks and quotes prices in such a way that it is
optimal to meet an exogenous demand for stocks. As we have proved
in~\cite{KramkovPulido:16}, the resulting stock prices can be
characterized in terms of solutions to a system of quadratic BSDEs
parametrized by the demand process.

It has been shown in~\cite{Germ:11} that under simple demands the
stock prices exist, are unique and can be constructed explicitly by
backward induction and martingale representation. For general
(nonsimple) demands the situation is more involved. In this case, the
existence and uniqueness of prices can be obtained only if the product
of certain model parameters is sufficiently small as in~\eqref{eq:24}
below. A natural question to ask is whether under such constraint the
output stock prices are stable under demands and, in particular,
whether they can be well approximated by the prices originated from
simple demands. A positive answer is given in Theorem~\ref{th:3} and
relies on the general stability estimates from Theorem~\ref{th:1}.

As the dealer's risk aversion coefficient $a$ approaches zero, the
price impact effect vanishes and we arrive at a classical model of
Mathematical Finance.  In Theorem~\ref{th:4} we derive an analytic
expansion of prices for sufficiently small values of $a$, thus getting
a natural scale of price impact corrections. The leading term of these
corrections has been obtained in \cite{Germ:11} for simple demands.

\subsection*{Notations}
\label{sec:notations-1}
For a matrix $A = (A^{ij})$ we denote its transpose by $A^*$ and
define its norm as
\begin{displaymath}
  \abs{A} \set  \sqrt{\trace{AA^*}} = \sqrt{\sum_{i,j}\abs{A^{ij}}^2}.
\end{displaymath}

We will work on a filtered probability space $(\Omega, \mathcal{F},
(\mathcal{F}_t)_{t\in [0,T]}, \mathbb{P})$ satisfying the standard
conditions of right-continuity and completeness; the initial
$\sigma$-algebra $\mathcal{F}_0$ is trivial,
$\mathcal{F}=\mathcal{F}_T$, and the maturity $T$ is finite. The
expectation is denoted as $\mathbb{E}[\cdot]$ and the conditional
expectation with respect to $\mathcal{F}_t$ as $\mathbb{E}_t[\cdot]$.

For an $n$-dimensional integrable random variable $\xi$ set
\begin{align*}
  \Lnorm{p}{\xi} \set& (\mathbb{E}[\abs{\xi}^p])^{1/p}, \quad p\geq 1, \\
  \Lnorm{\infty}{\xi} \set& \inf\descr{c\geq 0}{\abs{\xi(\omega)} \leq
    c, \quad \mathbb{P}[d\omega]-a.s.}.
\end{align*}

We shall use the following spaces of stochastic processes:
\begin{description}
\item[$\bmo(\mathbf{R}^n)$] is the Banach space of continuous
  $n$-dimensional martingales $M$ with $M_0=0$ and the norm
  \begin{displaymath}
    \snorm{\bmo}{M} \set  \sup_{\tau}
    \Lnorm{\infty}{\braces{\mathbb{E}_\tau[\qvar{M}_T-
        \qvar{M}_\tau]}^{1/2}},  
  \end{displaymath}
  where the supremum is taken with respect to all stopping times
  $\tau$ and $\qvar{M}$ is the quadratic variation of $M$.
\item[$\Sspace{\bmo}(\mathbf{R}^n)$] is the Banach space of continuous
  $n$-dimensional semimartingales $X=X_0+M+A$, where $M$ is a
  continuous martingale and $A$ is a process of finite variation, with
  the norm
  \begin{displaymath}
    \Sbmonorm{X}\set|X_0|+\snorm{\bmo}{M}
    +\sup_{\tau}\Lnorm{\infty}{\mathbb{E}_{\tau}[\int_\tau^T \abs{dA}]}.
  \end{displaymath}
  Here the supremum is taken over all stopping times $\tau$ and $\int
  \abs{dA}$ is the total variation of $A$.
\item[$\Sspace{p}(\mathbf{R}^n)$] for $p\geq 1$ is the Banach space of
  continuous $n$-dimensional semimartingales $X=X_0+M+A$, where $M$ is
  a continuous martingale and $A$ is a process of finite variation,
  with the norm
  \begin{displaymath}
    \Snorm{p}{X}\set|X_0|+ \Lnorm{p}{\qvar{M}^{1/2}_T} + \Lnorm{p}{\int_0^T
      \abs{dA}}.
  \end{displaymath}
\item[$\Hspace{0}(\mathbf{R}^{n\times d})$] is the vector space of
  predictable processes $\zeta$ with values in $n\times d$-matrices
  such that $\int_0^T \abs{\zeta_s}^2 ds < \infty$. This is precisely
  the space of $n\times d$-dimensional integrands $\zeta$ for a
  $d$-dimensional Brownian motion $B$. We shall identify $\alpha$ and
  $\beta$ in $\Hspace{0}(\mathbf{R}^{n\times d})$ if $\int_0^T
  \abs{\alpha_s- \beta_s}^2 ds = 0$ or, equivalently, if the
  stochastic integrals $\alpha\cdot B$ and $\beta\cdot B$ coincide.
\item[$\Hspace{p}(\mathbf{R}^{n\times d})$] for $p\geq 1$ consists of
  $\zeta \in \Hspace{0}(\mathbf{R}^{n\times d})$ such that $\zeta
  \cdot B \in \Sspace{p}(\mathbf{R}^n)$ for a $d$-dimensional Brownian
  motion $B$. It is a Banach space under the norm:
  \begin{displaymath}
    \Hnorm{p}{\zeta} \set \Snorm{p}{\zeta \cdot B}=
    \braces{\mathbb{E}\left[\left(\int_0^T
          \abs{\zeta_s}^2ds\right)^{p/2}\right]}^{1/p}.
  \end{displaymath}
\item[$\Hbmo(\mathbf{R}^{n\times d})$] consists of $\zeta \in
  \Hspace{0}(\mathbf{R}^{n\times d})$ such that $\zeta \cdot B \in
  \bmo(\mathbf{R}^n)$ for a $d$-dimensional Brownian motion $B$. It is
  a Banach space under the norm:
  \begin{displaymath}
    \Hnorm{\bmo}{\zeta} \set \snorm{\bmo}{\zeta \cdot B} = \sup_{\tau}
    \Lnorm{\infty}{\braces{\mathbb{E}_\tau\left[\int_\tau^T
          \abs{\zeta_s}^2ds\right]}^{1/2}}.
  \end{displaymath}
\item[$\Hspace{\infty}(\mathbf{R}^{n})$] is the Banach space of
  bounded $n$-dimensional predictable processes $\gamma$ with the
  norm:
  \begin{displaymath}
    \Hnorm{\infty}{\gamma} \set \inf\descr{c\geq
      0}{\abs{\gamma_t(\omega)} \leq c, \quad dt\times
      \mathbb{P}[d\omega]-a.s.}.
  \end{displaymath}
\end{description}

For an $n$-dimensional integrable random variable $\xi$ with
$\mathbb{E}[\xi] = 0$ set
\begin{displaymath}
  \Lnorm{\bmo}{\xi} \set \snorm{\bmo}{(\mathbb{E}_t[\xi])_{t\in
      [0,T]}}. 
\end{displaymath}

\section{Stability estimates}
\label{sec:stability_estimates}

Hereafter, we shall assume that
\begin{enumerate}[label=(A\arabic{*}), ref=(A\arabic{*})]
\item \label{item:1} There exists a $d$-dimensional Brownian motion
  $B$ such that every local martingale $M$ is a stochastic integral
  with respect to $B$:
  \begin{displaymath}
    M =  M_0 + \zeta \cdot B.
  \end{displaymath}
  \setcounter{item}{\value{enumi}}
\end{enumerate}
Of course, this assumption holds if the filtration is generated by
$B$.

Consider the $n$-dimensional BSDE:
\begin{equation}
  \label{eq:1}
  Y_t=\Xi + \int_t^Tf(s,\zeta_s)\,ds-\int_t^T\zeta\,dB,\quad t\in
  [0,T].
\end{equation}
Here $Y$ is an $n$-dimensional semimartingale, $\zeta$ is a
predictable process with values in the space of $n\times d$ matrices,
and the terminal condition $\Xi$ and the driver $f=f(t,z)$ satisfy the
following assumptions:
\begin{enumerate}[label=(A\arabic{*}), ref=(A\arabic{*})]
  \setcounter{enumi}{\value{item}}
\item \label{item:2} $\Xi$ is an integrable random variable with
  values in $\mathbf{R}^n$ such that the martingale
  \begin{displaymath}
    L_t \set \mathbb{E}_t[\Xi] - \mathbb{E}[\Xi], \quad t\in [0,T],
  \end{displaymath}
  belongs to $\bmo$.
\item \label{item:3} $t\mapsto f(t,z)$ is a predictable process with
  values in $\mathbf{R}^n$,
  \begin{equation}
    \label{eq:32}
    f(t,0)=0,
  \end{equation}
  and there is a constant $\Theta>0$ such that
  \begin{equation}
    \label{eq:31}
    |f(t,u)-f(t,v)| \leq
    \Theta(|u-v|)(|u|+|v|),
  \end{equation}
  for all $t\in [0,T]$ and $u,v\in \mathbf{R}^{n\times d}$.
  \setcounter{item}{\value{enumi}}
\end{enumerate}
Note that $f=f(t,z)$ has a quadratic growth in $z$. We shall
discuss~\ref{item:3} in Remark~\ref{rem:3} at the end of this section.

Recall that there is a constant $\kappa = \kappa(n)$ such that, for
every martingale $M\in \bmo(\mathbf{R}^n)$,
\begin{equation}
  \label{eq:2}
  \frac1{\kappa} \snorm{\bmo}{M} \leq \snorm{\bmo_1}{M} \set
  \sup_{\tau}\Lnorm{\infty}{\mathbb{E}_\tau[\abs{M_T-M_\tau}]} \leq
  \snorm{\bmo}{M}, 
\end{equation}
see~\cite{Kazam:94}, Corollary~2.1.  Theorem A.1
in~\cite{KramkovPulido:16} shows that under \ref{item:1},
\ref{item:2}, and \ref{item:3}, if
\begin{equation}
  \label{eq:33}
  \snorm{\bmo}{L} < \frac1{8\kappa\Theta},
\end{equation}
then there exists a unique solution $(Y,\zeta)$ to~\eqref{eq:1} such
that
\begin{displaymath}
  \Hnorm{\bmo}{\zeta} \leq  \frac1{4\kappa\Theta}. 
\end{displaymath}
The analogous local existence and uniqueness result was first shown
in~\cite[Proposition 1]{Tevz:08} for small bounded terminal conditions
and then in~\cite[Proposition 2.1]{Frei:14} in the current $\bmo$
setting, but with different constants.

Theorem~\ref{th:1} below provides stability estimates for such
\emph{local} solution $(Y,\zeta)$ with respect to the terminal
condition and the driver. Along with~\eqref{eq:1}, we consider a
similar $n$-dimensional BSDE
\begin{align}
  \label{eq:3}
  Y'_t =\Xi' + \int_t^Tf'(s,\zeta'_s)\,ds-\int_t^T\zeta'\,dB,\quad
  t\in [0,T],
\end{align}
whose terminal condition $\Xi'$ and the driver $f' = f'(t,z)$ satisfy
the same conditions~\ref{item:2} and \ref{item:3} as $\Xi$ and $f$. We
denote
\begin{displaymath}
  L'_t \set \mathbb{E}_t[\Xi'] - \mathbb{E}[\Xi'], \quad t\in [0,T],
\end{displaymath}
and assume that there exists a nonnegative process $\delta=(\delta_t)$
such that
\begin{equation}
  \label{eq:4}
  \abs{f(t,z)-f'(t,z)} \leq \delta_t\abs{z}^2.
\end{equation}

\begin{Theorem}
  \label{th:1}
  Assume that the BSDEs~\eqref{eq:1} and~\eqref{eq:3}
  satisfy~\ref{item:1}, \ref{item:2}, \ref{item:3}, and~\eqref{eq:4}
  and let $(Y,\zeta)$ and $(Y',\zeta')$ be their respective
  solutions. For $p>1$ there are positive constants $c=c(n,p)$ and
  $C=C(n,p)$ (depending only on $n$ and $p$) such that if
  \begin{equation}
    \label{eq:5}
    \Hnorm{\bmo}{\zeta} + \Hnorm{\bmo}{\zeta'}  \leq \frac{c}{\Theta},
  \end{equation}
  then
  \begin{align}
    \label{eq:6}
    \Hnorm{p}{\zeta'- \zeta} &\leq C \left(\Lnorm{p}{L'_T-L_T} +
      \Hnorm{2p}{\sqrt{\delta}\zeta}^2\right), \\
    \label{eq:7}
    \Snorm{p}{Y'- Y} &\leq C \left(\Lnorm{p}{\Xi'-\Xi} +
      \Hnorm{2p}{\sqrt{\delta}\zeta}^2\right).
  \end{align}
  Moreover, there exist constants $\widetilde{c}=\widetilde{c}(n)$ and
  $\widetilde{C}=\widetilde{C}(n)$ (depending only on $n$) such that
  if
  \begin{displaymath}
    \Hnorm{\bmo}{\zeta} + \Hnorm{\bmo}{\zeta'}  \leq \frac{\widetilde{c}}{\Theta},
  \end{displaymath}
  then
  \begin{align}
    \label{eq:6-1}
    \Hnorm{\bmo}{\zeta'- \zeta} &\leq \widetilde{C}
    \left(\snorm{\bmo}{L'-L} +
      \Hnorm{\bmo}{\sqrt{\delta}\zeta}^2\right), \\
    \label{eq:7-1}
    \Snorm{\bmo}{Y'- Y} &\leq \widetilde{C}
    \left(\abs{\mathbb{E}[\Xi'-\Xi]}+ \snorm{\bmo}{L'-L} +
      \Hnorm{\bmo}{\sqrt{\delta}\zeta}^2\right).
  \end{align}
\end{Theorem}

\begin{proof}
  To shorten the notations set $\Delta \zeta \set \zeta' - \zeta$,
  etc.  Define the martingale $M$ and the process $A$ of bounded
  variation as
  \begin{align*}
    M  &\set \Delta \zeta \cdot B, \\
    A_t &\set \int_0^t (f'(s,\zeta'_s) - f(s,\zeta_s)) ds, \quad t\in
    [0,T].
  \end{align*}
  We deduce that
  \begin{displaymath}
    M_T = \Delta L_T + A_T - \mathbb{E}[A_T],
  \end{displaymath}
  which readily implies that
  \begin{displaymath}
    \Lnorm{p}{M_T} \leq \Lnorm{p}{\Delta L_T} +  2\Lnorm{p}{A_T}.
  \end{displaymath}
  From Doob's and Burkholder-Davis-Gundy's (BDG) inequalities we
  deduce the existence of a constant $C_1 = C_1(n,p)$ such that
  \begin{displaymath}
    \Hnorm{p}{\Delta \zeta} \leq C_1\Lnorm{p}{M_T}.
  \end{displaymath}

  To estimate $\Lnorm{p}{A_T}$ we use the potential $Z$ associated
  with the variation of $A$:
  \begin{align*}
    Z_t &\set \mathbb{E}_t\left[\int_t^T \abs{dA}\right] =
    \mathbb{E}_t\left[\int_t^T \abs{f'(s,\zeta'_s) -
        f(s,\zeta_s)}ds\right] \\
    & \leq \mathbb{E}_t\left[\int_t^T \abs{f'(s,\zeta'_s) -
        f'(s,\zeta_s)}ds\right] +
    \mathbb{E}_t\left[\int_t^T\abs{f'(s,\zeta_s) - f(s,\zeta_s)}
      ds\right].
  \end{align*}
  Denoting $Z^* \set \sup_{t\in [0,T]}\abs{Z_t}$ we have, by the
  Garsia-Neveu inequality,
  \begin{displaymath}
    \Lnorm{p}{A_T} \leq \Lnorm{p}{\int_0^T \abs{dA}} \leq p
    \Lnorm{p}{Z^*}. 
  \end{displaymath}

  Take $1<p'<p$ and denote $q' \set p'/(p'-1)$. From~\ref{item:3} and
  the Cauchy-Schwarz and H\"older inequalities we obtain
  \begin{align*}
    U_t & \set \mathbb{E}_t\left[\int_t^T \abs{f'(s,\zeta'_s) -
        f'(s,\zeta_s)} ds\right] \leq \Theta \mathbb{E}_t\left[
      \int_t^T (\abs{\zeta_s} + \abs{\zeta'_s})\abs{\Delta \zeta_s}
      ds\right] \\
    & \leq \Theta\mathbb{E}_t\left[\left(\int_t^T(\abs{\zeta_s} +
        \abs{\zeta'_s})^2 ds\right)^{1/2} \left(\int_t^T\abs{\Delta
          \zeta_s}^2
        ds \right)^{1/2}\right] \\
    &\leq \Theta \left(\mathbb{E}_t\left[\left(\int_t^T(\abs{\zeta_s}
          + \abs{\zeta'_s})^2 ds\right)^{q'/2}\right]\right)^{1/q'}
    \left(\mathbb{E}_t\left[\left(\int_t^T\abs{\Delta \zeta_s}^2
          ds\right)^{p'/2}\right]\right)^{1/p'}.
  \end{align*}
  From Doob's inequality, the BDG inequalities, and the equivalence of
  $\bmo_p$-norms, see~\cite[Corollary~2.1, p.~28]{Kazam:94}, we deduce
  the existence of a constant $C_2 = C_2(n,q')$ such that
  \begin{displaymath}
    \left(\mathbb{E}_t\left[\left(\int_t^T(\abs{\zeta_s}
          + \abs{\zeta'_s})^2 ds\right)^{q'/2}\right]\right)^{1/q'} \leq
    C_2 (\Hnorm{\bmo}{\zeta} + \Hnorm{\bmo}{\zeta'}).
  \end{displaymath}
  Using the obvious estimate
  \begin{displaymath}
    \mathbb{E}_t\left[\left(\int_t^T\abs{\Delta \zeta_s}^2
        ds\right)^{p'/2}\right] \leq
    \mathbb{E}_t\left[\left(\int_0^T\abs{\Delta \zeta_s}^2
        ds\right)^{p'/2}\right] \set N_t,
  \end{displaymath}
  and denoting $r\set p/p'>1$, we deduce from Doob's inequality the
  existence of a constant $C_3 = C_3(r)$ such that
  \begin{align*}
    \Lnorm{r}{N^*} \leq C_3 \Lnorm{r}{N_T} = C_3 \Hnorm{p}{\Delta
      \zeta}^{p'}.
  \end{align*}
  Combining the above estimates we obtain
  \begin{displaymath}
    \Lnorm{p}{U^*} \leq C_4 \Theta (\Hnorm{\bmo}{\zeta} + \Hnorm{\bmo}{\zeta'})
    \Hnorm{p}{\Delta \zeta},
  \end{displaymath}
  where the constant $C_4 = C_2 C_3^{{1}/{p'}}$ depends only on $n$
  and $p$.

  From~\eqref{eq:4} we deduce
  \begin{align*}
    V_t &\set \mathbb{E}_t\left[ \int_t^T \abs{f'(s,\zeta_s) -
        f(s,\zeta_s)} ds\right] \leq \mathbb{E}_t\left[ \int_t^T
      \delta_s
      \abs{\zeta_s}^2 ds\right] \\
    &\leq \mathbb{E}_t\left[\int_0^T \delta_s \abs{\zeta_s}^2
      ds\right].
  \end{align*}
  Another application of Doob's inequality yields a constant $C_5 =
  C_5(p)>1$ such that
  \begin{displaymath}
    \Lnorm{p}{V^*} \leq  C_5 \Lnorm{p}{\int_0^T \delta_s \abs{\zeta_s}^2
      ds}. 
  \end{displaymath}

  Defining the constants $c=c(n,p)$ and $C=C(n,p)$ as
  \begin{align*}
    c = \frac1{4pC_1 C_4}, \quad C = 4pC_1C_5,
  \end{align*}
  and assuming~\eqref{eq:5}, we obtain
  \begin{align*}
    \Hnorm{p}{\Delta \zeta} &\leq C_1p (\Lnorm{p}{\Delta L_T} + 2\Lnorm{p}{Z^*}) \\
    &\leq C_1p
    (\Lnorm{p}{\Delta L_T} + 2\Lnorm{p}{U^*} + 2\Lnorm{p}{V^*}) \\
    &\leq C_1p \Bigl(\Lnorm{p}{\Delta L_T} + 2C_4
    \Theta(\Hnorm{\bmo}{\zeta} + \Hnorm{\bmo}{\zeta'})
    \Hnorm{p}{\Delta \zeta} \\
    &\qquad + 2C_5 \Lnorm{p}{\int_0^T \delta_s
      \abs{\zeta_s}^2 ds}\Bigr) \\
    &\leq \frac{1}2 \Hnorm{p}{\Delta \zeta} + \frac{C}2
    \Bigl(\Lnorm{p}{\Delta L_T} +
    \Hnorm{2p}{\sqrt{\delta}\zeta}^2\Bigr),
  \end{align*}
  which implies~\eqref{eq:6}.

  The estimate~\eqref{eq:7} follows from~\eqref{eq:6} and estimates
  above for $\Lnorm{p}{\int_0^T \abs{dA}}$ with appropriate $C=C(n,p)$
  as soon as we write
  \begin{align*}
    \Delta Y_0 &= \mathbb{E}[\Delta \Xi + A_T], \\
    \Delta Y &= \Delta Y_0 + M - A,
  \end{align*}
  with $M$ and $A$ defined at the beginning of the proof.

  Estimates~\eqref{eq:6-1} and~\eqref{eq:7-1} are more
  straightforward. Using the same processes $M$, $A$, $U$, and $V$ we
  have that
  \begin{displaymath}
    M_t=\Delta L_t+\E_t[A_T]-\E[A_T]
  \end{displaymath}
  and for a stopping time $\tau$
  \begin{displaymath}
    \E_{\tau}[\abs{A_T-\E_{\tau}[A_T]}] \leq 2 \mathbb{E}_\tau\left[\int_\tau^T \abs{f'(s,\zeta'_s) -
        f(s,\zeta_s)} ds\right] \leq 2(U_\tau + V_{\tau}). 
  \end{displaymath}
  By the Cauchy-Schwarz inequality we deduce from~\ref{item:3} that
  \begin{align*}
    U_\tau &= \mathbb{E}_\tau\left[\int_\tau^T \abs{f'(s,\zeta'_s) -
        f'(s,\zeta_s)} ds\right] \leq \Theta \mathbb{E}_\tau\left[
      \int_\tau^T (\abs{\zeta_s} + \abs{\zeta'_s})\abs{\Delta \zeta_s}
      ds\right] \\ &\leq \Theta(\Hnorm{\bmo}{\zeta} +
    \Hnorm{\bmo}{\zeta'}) \Hnorm{\bmo}{\Delta\zeta}
  \end{align*}
  and from~\eqref{eq:4} that
  \begin{align*}
    V_\tau &= \mathbb{E}_\tau\left[\int_\tau^T \abs{f'(s,\zeta_s) -
        f(s,\zeta_s)} ds\right] \leq \mathbb{E}_\tau\left[ \int_\tau^T
      \delta_s
      \abs{\zeta_s}^2 ds\right] \\
    &\leq \Hnorm{\bmo}{\sqrt{\delta}\zeta}^2.
  \end{align*}
  Accounting for~\eqref{eq:2} we obtain that
  \begin{align*}
    \frac1{2\kappa}\Lnorm{\bmo}{A_T-\E[A_T]} \leq
    \Theta(\Hnorm{\bmo}{\zeta} + \Hnorm{\bmo}{\zeta'} )
    \Hnorm{\bmo}{\Delta\zeta} + \Hnorm{\bmo}{\sqrt{\delta}\zeta}^2
  \end{align*}
  and since
  \begin{align*}
    \Hnorm{\bmo}{\Delta\zeta}=\snorm{\bmo}{M}\leq \snorm{\bmo}{\Delta
      L} + \Lnorm{\bmo}{A_T-\E[A_T]},
  \end{align*}
  estimate~\eqref{eq:6-1} readily follows if we choose
  \begin{align*}
    \widetilde{c} = \frac1{4\kappa}, \quad \widetilde{C} = 4\kappa.
  \end{align*}

  Finally, we deduce~\eqref{eq:7-1} (with appropriate
  $\widetilde{C}=\widetilde{C}(n)$) from the estimates for
  $\snorm{\bmo}{M}$ and $\Lnorm{\bmo}{A_T-\E[A_T]}$ as soon as we
  write
  \begin{displaymath}
    \Delta Y = \E[\Delta \Xi] + M - (A-\E[A_T]).
  \end{displaymath}
\end{proof}

\begin{Remark}
  Assume that BSDEs~\eqref{eq:1} and~\eqref{eq:3}
  satisfy~\ref{item:1}, \ref{item:2}, \ref{item:3}, and~\eqref{eq:4}
  and let $(Y,\zeta)$ and $(Y',\zeta')$ be their respective
  solutions. If we set
  \begin{displaymath}
    \begin{split}
      A_t&\set\int_0^t (f'(s,\zeta_s)-f(s,\zeta_s)) ds,\\
      \gamma_t&\set \frac{1}{\abs{\zeta'_t-\zeta_t}}(f'(s,\zeta'_s) -
      f'(s,\zeta_s)){\rm 1}_{\{\abs{\zeta'_t-\zeta_t}>0\}},
    \end{split}
  \end{displaymath}
  then we can write
  \begin{displaymath}
    Y'_t-Y_t=\Xi'-\Xi+A_T-A_t+\int_t^T\gamma_s\abs{\zeta'_s-\zeta_s}
    ds- \int_t^T(\zeta'-\zeta)\,dB. 
  \end{displaymath}
  This linearization argument allows us to deduce~\eqref{eq:7} from
  the stability estimate for Lipschitz BSDEs given in Theorem 2.2
  of~\cite{DelbTang:08} if we observe that the local
  bound~\eqref{eq:5} implies the conditions of item~(ii) of this
  theorem. We thank a referee for pointing out this connection.
\end{Remark}

\begin{Remark}
  \label{rem:3}
  Assumption~\ref{item:3} on the driver $f=f(t,z;\omega)$ allows us to
  obtain existence of solutions to~\eqref{eq:1} under the
  condition~\eqref{eq:33} which does \emph{not} involve the maturity
  $T$. Such $T$-independence property is not valid
  without~\eqref{eq:32} (this fact is easy to see) or if~\eqref{eq:31}
  is replaced with a weaker condition:
  \begin{equation}
    \label{eq:36}
    |f(t,u)-f(t,v)| \leq \Theta(|u-v|)(1+|u|+|v|).
  \end{equation}
  We construct below a related counterexample using the same idea as
  in \cite{FreiReis:11}.

  We take a one-dimensional Brownian motion $B$ and define the
  stopping time $\tau$ and the bounded martingale $M$ as
  \begin{align*}
    \tau &\set \min\descr{t\in (0,1)}{\abs{\int_0^t \frac1{1-s}dB_s} =
      \frac{\pi}2}, \\
    M_t &\set \int_0^{t\wedge \tau} \frac1{1-s}dB_s, \quad t\geq 0.
  \end{align*}
  We deduce that $\tau \in (0,1)$ and $\abs{M_\tau} = \pi/2$.
  Standard arguments (see \cite[Lemma~1.3]{Kazam:94}) show that
  \begin{equation}
    \label{eq:34}
    \mathbb{E}\left[\exp\left(\frac{a^2}2 \qvar{M}_\tau\right)\right] = \left\{
      \begin{array}[c]{cc}
        1/{\cos(a\pi/2)}, &\quad 0\leq a< 1, \\
        \infty, &\quad a\geq 1. 
      \end{array}
    \right.
  \end{equation}

  We consider the three-dimensional quadratic BSDE, parametrized by
  $a>0$ and $T>1$:
  \begin{align*}
    Y^1_t & = \frac{a}{\sqrt{T}} B_T - \int_t^T \zeta^1 dB, \\
    Y^2_t & = \int_t^T M_\tau \ind{s\geq 1} \zeta^1_s ds - \int_t^T
    \zeta^2 dB, \\
    Y^3_t & = \frac12 \int_t^T \left((\zeta^2_s)^2 +
      (\zeta^3_s)^2\right) ds - \int_t^T \zeta^3 dB.
  \end{align*}
  We notice that the driver in the second equation depends on
  $\zeta^1$ linearly, that~\eqref{eq:36} holds with
  $\Theta=\frac{\pi}2$ and that
  \begin{equation}
    \label{eq:35}
    \Lnorm{\bmo}{\Xi} =
    \Lnorm{\bmo}{\bigl(\frac{a}{\sqrt{T}}B_T,0,0\bigr)} = a.  
  \end{equation}

  If $\int \zeta^1 dB$ and $\int \zeta^2 dB$ are true martingales on
  $[0,T]$, then the first two equations yield that
  \begin{align*}
    \zeta^1 &= \frac{a}{\sqrt{T}}, \\
    \int_0^T \zeta^2 dB &= \frac{a(T-1)}{\sqrt{T}} M_\tau.
  \end{align*}
  The third equation implies that
  \begin{align*}
    \exp\left(Y^3_0 + \int_0^T \zeta^3 dB - \frac12 \int_0^T
      (\zeta^3_s)^2
      ds\right) &= \exp\left(\frac12 \int_0^T (\zeta^2_s)^2 ds\right) \\
    &= \exp\left(\frac{a^2(T-1)^2}{2T} \qvar{M}_\tau \right).
  \end{align*}
  In view of~\eqref{eq:34} and~\eqref{eq:35}, the existence of
  $\zeta^3$ is then equivalent to
  \begin{align*}
    \frac{a(T-1)}{\sqrt{T}} = \frac{\Lnorm{\bmo}{\Xi}(T-1)}{\sqrt{T}}
    < 1.
  \end{align*}
  Thus, the solvability of our BSDE depends on \emph{both}
  $\Lnorm{\bmo}{\Xi}$ and $T$.
\end{Remark}

\section{Analytic expansion for purely quadratic BSDE}
\label{sec:asympt-analys-quadr}

Consider an $n$-dimensional BSDE
\begin{equation}
  \label{eq:8}
  Y_t=a\Xi + \int_t^Tf(s,\zeta_s)\,ds-\int_t^T\zeta\,dB,\quad t\in
  [0,T],
\end{equation}
where the terminal condition depends on a parameter
$a\in\mathbf{R}$. If $\Xi$ and $f$ satisfy~\ref{item:2}
and~\ref{item:3} and $\abs{a}<\rho$, where
\begin{equation}
  \label{eq:9}
  \rho\set  \frac1{8\kappa\Theta\snorm{\bmo}{L}},
\end{equation}
then, by~\cite[Theorem A.1]{KramkovPulido:16}, there is only one
solution $(Y(a),\zeta(a))$ such that
\begin{equation}
  \label{eq:10}
  \Hnorm{\bmo}{\zeta(a)} \leq  \frac1{4\kappa\Theta},
\end{equation}
and for this solution we have an estimate:
\begin{displaymath}
  \Hnorm{\bmo}{\zeta(a)} \leq 2|a| \snorm{\bmo}{L}.
\end{displaymath}
In particular, $\zeta(a)$ converges to 0 in $\Hbmo$ as $a$ approaches
0.

In Theorem~\ref{th:2} below we obtain an analytic expansion for
$\zeta(a)$ in the neighborhood of $a=0$ under the additional
assumption that the driver $f=f(t,z)$ is purely quadratic in $z$:
\begin{displaymath}
  f(t,z) = \sum_{ijkl}z^{ij}z^{kl} \alpha^{ijkl}_t
\end{displaymath}
for some $\mathbf{R}^n$-valued predictable bounded processes
$(\alpha^{ijkl})$ or, equivalently,
\begin{enumerate}[label=(A\arabic{*}), ref=(A\arabic{*})]
  \setcounter{enumi}{\value{item}}
\item \label{item:4} $f(t,z)=\widetilde{f}(t,z,z)$, where for all $u,v
  \in \mathbf{R}^{n\times d}$ the map $t\mapsto \widetilde{f}(t,u,v)$
  is an $\mathbf{R}^n$-valued predictable process and for every $t\in
  [0,T]$ the map $(u,v)\mapsto \widetilde f(t,u,v)$ is symmetric,
  bilinear on $\mathbf{R}^{n\times d}\times \mathbf{R}^{n\times d}$,
  and bounded by a constant $\Theta>0$. In other words,
  \begin{align*}
    \widetilde{f}(t,\lambda u, v+w) &= \widetilde{f}(t,\lambda
    (v+w),u) = \lambda (\widetilde{f}(t,u,v) + \widetilde{f}(t,u,w)),
    \\ \abs{\widetilde{f}(t,u,v)} &\leq \Theta \abs{u} \abs{v},
  \end{align*}
  for all $t\in [0,T]$, $\lambda \in \mathbf{R}$, and $u,v,w\in
  \mathbf{R}^{n\times d}$.  \setcounter{item}{\value{enumi}}
\end{enumerate}

Notice that~\ref{item:4} implies~\ref{item:3}:
\begin{align*}
  |f(t,u)-f(t,v)| &= \abs{\widetilde{f}(t,u,u) - \widetilde{f}(t,v,v)}
  = \abs{\widetilde{f}(t,u-v,u+v)} \\
  &\leq \Theta(|u-v|)(|u+v|) \leq \Theta(|u-v|)(|u|+|v|).
\end{align*}
Condition~\ref{item:4} naturally arises in financial applications
dealing with exponential utilities, see e.g., \cite{FreiReis:11},
\cite{KramkovPulido:16}, and \cite{KardXingZitk:15}.

To state Theorem~\ref{th:2} we need the following technical result.

\begin{Lemma}
  \label{lem:1}
  Assume \ref{item:1} and~\ref{item:4}. For $\mu,\nu\in\Hbmo$ there is
  a unique $\zeta\in\Hbmo$ such that
  \begin{equation}
    \label{eq:11}
    (\zeta\cdot B)_t = \mathbb{E}_t\left[
      \int_0^T\widetilde{f}(s,\mu_s,\nu_s)\,ds\right]
    -\mathbb{E}\left[\int_0^T\widetilde{f}(s,\mu_s,\nu_s)\,ds
    \right].
  \end{equation}
  Moreover,
  \begin{displaymath}
    \Hnorm{\bmo}{\zeta} \leq 2\kappa\Theta\Hnorm{\bmo}{\mu}\Hnorm{\bmo}{\nu},
  \end{displaymath}
  where the positive constants $\kappa$ and $\Theta$ are defined
  in~\eqref{eq:2} and~\ref{item:4}.
\end{Lemma}
\begin{proof}
  Define the martingale
  \begin{displaymath}
    M_t \set  \mathbb{E}_t\left[
      \int_0^T\widetilde{f}(s,\mu_s,\nu_s)\,ds\right]
    -\mathbb{E}\left[\int_0^T\widetilde{f}(s,\mu_s,\nu_s)\,ds
    \right].
  \end{displaymath}
  For a stopping time $\tau$ we deduce from~\ref{item:4} and the
  Cauchy-Schwarz inequality that
  \begin{align*}
    \mathbb{E}_{\tau}[\abs{M_T - M_\tau}] & =
    \mathbb{E}_\tau\left[\abs{\int_\tau^T
        \widetilde{f}(s,\mu_s,\nu_s)ds - \mathbb{E}_\tau[\int_\tau^T
        \widetilde{f}(s,\mu_s,\nu_s)ds]}\right] \\
    & \leq
    2\mathbb{E}_{\tau}\left[\int_{\tau}^T\abs{\widetilde{f}(s,\mu_s,\nu_s)}
      ds\right] \leq
    2\Theta\mathbb{E}_{\tau}\left[\int_{\tau}^T\abs{\mu_s} \abs{\nu_s}ds\right] \\
    &\leq
    2\Theta\left(\mathbb{E}_{\tau}\left[\int_{\tau}^T\abs{\mu_s}^2
        ds\right]\right)^{1/2}
    \left(\mathbb{E}_{\tau}\left[\int_{\tau}^T\abs{\nu_s}^2
        ds\right]\right)^{1/2}.
  \end{align*}
  Hence by~\eqref{eq:2},
  \begin{displaymath}
    \snorm{\bmo}{M} \leq \kappa \sup_{\tau}
    \Lnorm{\infty}{\mathbb{E}_{\tau}[\abs{M_T - 
        M_\tau}] } \leq 2\kappa\Theta
    \Hnorm{\bmo}{\mu}\Hnorm{\bmo}{\nu},  
  \end{displaymath}
  and the result follows because, in view of~\ref{item:1}, $M$ admits
  an integral representation $\zeta\cdot B$ for some unique $\zeta\in
  \Hbmo$.
\end{proof}

Lemma~\ref{lem:1} allows us to define the map
\begin{displaymath}
  \map{\widetilde F}{\Hbmo\times \Hbmo}{\Hbmo}
\end{displaymath}
such that $\zeta= \widetilde{F}(\mu,\nu)$ is given
by~\eqref{eq:11}. This map is bilinear (since
$\widetilde{f}(t,\cdot,\cdot)$ is bilinear) and is bounded by
$2\kappa\Theta$.

Recall the constant $\rho$ from~\eqref{eq:9} and the $\bmo$ martingale
$L$ from~\ref{item:2}.

\begin{Theorem}
  \label{th:2}
  Assume ~\ref{item:1}, \ref{item:2}, and~\ref{item:4}. Then for
  $\abs{a}<\rho$ there is only one solution $(Y(a),\zeta(a))$
  to~\eqref{eq:8} satisfying~\eqref{eq:10}.  It is given by the power
  series
  \begin{displaymath}
    Y(a) =\sum_{k=1}^\infty Y^{(k)} a^k \quad\text{and} \quad
    \zeta(a) =\sum_{k=1}^\infty \zeta^{(k)} a^k
  \end{displaymath}
  convergent for $\abs{a}<\rho$ in $\Sbmo$ and $\Hbmo$, respectively,
  with the coefficients
  \begin{align}
    \label{eq:12}
    Y_t^{(1)} & = \mathbb{E}_t[\Xi], \quad t\in [0,T], \\
    \label{eq:13}
    \zeta^{(1)}\cdot B &= L,
  \end{align}
  and, for $k\geq 2$,
  \begin{align}
    \label{eq:14}
    \zeta^{(k)} &= \sum_{l+m=k}
    \widetilde{F}(\zeta^{(l)},\zeta^{(m)}), \\
    \label{eq:15}
    Y_t^{(k)} & = \sum_{l+m=k} \mathbb{E}_t[\int_t^T \widetilde
    f(s,\zeta_s^{(l)},\zeta_s^{(m)}) ds], \quad t\in [0,T],
  \end{align}
  where we sum with respect to all pairs of positive integers $(l,m)$
  which add to $k$.
\end{Theorem}

\begin{Remark}
  Based on~\eqref{eq:14}--\eqref{eq:15} and the definition of the
  bilinear map $\widetilde{F}$, the power series expansion of
  $(Y(a),\zeta(a))$ is constructed using a martingale representation
  approach. We refer the reader to~\cite{ContLu:16} and the references
  therein for numerical algorithms to approximate martingale
  representation terms of functionals of Brownian motion.
\end{Remark}

The proof of Theorem~\ref{th:2} relies on some lemmas.

\begin{Lemma}
  \label{lem:2}
  Assume ~\ref{item:1}, \ref{item:2}, and~\ref{item:4}. Let
  $(\zeta^{(k)})_{k\geq 1}$ be given by
  \eqref{eq:13}--\eqref{eq:14}. Then $(\zeta^{(k)})_{k\geq 1} \subset
  \Hbmo$ and
  \begin{equation}
    \label{eq:16}
    \sum_{k=1}^\infty \Hnorm{\bmo}{\zeta^{(k)}} \rho^k \leq
    \frac{1}{4\kappa\Theta}.
  \end{equation}
\end{Lemma}
\begin{proof}
  The claim that each $\zeta^{(k)}$ belongs to $\Hbmo$ follows from
  its construction and Lemma~\ref{lem:1}. For $n\geq 1$ define the
  partial sums
  \begin{displaymath}
    s_n\set \sum_{k=1}^n \Hnorm{\bmo}{\zeta^{(k)}} \rho^k.
  \end{displaymath}
  Using the boundedness of the bilinear map $\widetilde{F}$ by
  $2\kappa \Theta$ we obtain
  \begin{align*}
    s_n - s_1 &= \sum_{k=2}^n \Hnorm{\bmo}{\zeta^{(k)}}\rho^k \leq
    \sum_{k=2}^n \left(\sum_{l+m=k}
      \Hnorm{\bmo}{\widetilde{F}(\zeta^{(l)},\zeta^{(m)})}\right)
    \rho^k \\
    &\leq 2\kappa\Theta \sum_{k=2}^n \sum_{l+m=k}
    (\Hnorm{\bmo}{\zeta^{(l)}}\rho^l)(\Hnorm{\bmo}{\zeta^{(m)}}\rho^{m}) \\
    &\leq 2\kappa\Theta \sum_{l,m=1}^{n-1}
    (\Hnorm{\bmo}{\zeta^{(l)}}\rho^{l})(\Hnorm{\bmo}{\zeta^{(m)}}\rho^m)
    \\
    & = 2\kappa\Theta \left(\sum_{k=1}^{n-1}
      \Hnorm{\bmo}{\zeta^{(k)}}\rho^k\right)^2 = 2\kappa\Theta
    (s_{n-1})^2.
  \end{align*}

  To verify~\eqref{eq:16} we use an induction argument.  For $n=1$ we
  have
  \begin{displaymath}
    s_1 =  \Hnorm{\bmo}{\zeta^{(1)}} \rho =  \rho \snorm{\bmo}{L} =
    \frac{1}{8\kappa\Theta}.
  \end{displaymath}
  If now $s_{n-1} \leq 1/({4\kappa\Theta})$, then
  \begin{displaymath}
    s_n\leq s_1 + 2\kappa\Theta (s_{n-1})^2 \leq
    \frac{1}{8\kappa\Theta} + 2\kappa\Theta
    \left(\frac{1}{4\kappa\Theta}\right)^2 = \frac{1}{4\kappa\Theta}
  \end{displaymath}
  and~\eqref{eq:16} follows.
\end{proof}

\begin{Lemma}
  \label{lem:3}
  Assume \ref{item:1} and~\ref{item:4}. For $\mu,\nu\in\Hbmo$ the
  process
  \begin{displaymath}
    X_t \set \mathbb{E}_t\left[
      \int_t^T\widetilde{f}(s,\mu_s,\nu_s)\,ds\right], \quad t\in
    [0,T], 
  \end{displaymath}
  belongs to $\Sbmo$ and
  \begin{displaymath}
    \Sbmonorm{X} \leq 2(1+\kappa)\Theta\Hnorm{\bmo}{\mu}\Hnorm{\bmo}{\nu}. 
  \end{displaymath}
\end{Lemma}
\begin{proof} The canonical decomposition of the semimartingale $X$
  has the form
  \begin{displaymath}
    X = X_0 + M - A, 
  \end{displaymath} 
  where
  \begin{align*}
    A_t & = \int_0^t \widetilde{f}(s,\mu_s,\nu_s)\,ds, \\
    X_0 &= \mathbb{E}[A_T], \\
    M_t &= \mathbb{E}_t[A_T] - \mathbb{E}[A_T].
  \end{align*}
  By Lemma~\ref{lem:1} we have
  \begin{displaymath}
    \snorm{\bmo}{M} \leq 2 \kappa \Theta \Hnorm{\bmo}{\mu}
    \Hnorm{\bmo}{\nu}. 
  \end{displaymath}
  As in the proof of Lemma~\ref{lem:1} we deduce that for any stopping
  time $\tau$
  \begin{displaymath}
    \mathbb{E}_\tau[\int_\tau^T \abs{dA}] =
    \mathbb{E}_\tau[\int_\tau^T
    \abs{\widetilde{f}(s,\mu_s,\nu_s)}\,ds] \leq  \Theta \Hnorm{\bmo}{\mu}
    \Hnorm{\bmo}{\nu} 
  \end{displaymath}
  and the result readily follows.
\end{proof}

\begin{proof}[Proof of Theorem \ref{th:2}]
  Take $a\in\mathbf{R}$ such that $\abs{a}<\rho$. Recall that
  \ref{item:4} implies \ref{item:3}.  Theorem A.1
  in~\cite{KramkovPulido:16} then implies the existence and uniqueness
  of the solution $\zeta(a)$ satisfying~\eqref{eq:10}. To show that
  \begin{equation}
    \label{eq:17}
    \zeta(a) = \beta \set \sum_{k=1}^\infty \zeta^{(k)} a^k,
  \end{equation}
  we need to verify that $\beta$ is a fixed point of the map
  $\map{F}{\Hbmo}{\Hbmo}$ given by
  \begin{displaymath}
    F(\zeta) \set a\zeta^{(1)} + \widetilde F(\zeta,\zeta).
  \end{displaymath}

  For $n\geq 1$ define the partial sums:
  \begin{displaymath}
    \beta_n\set\sum_{k=1}^n \zeta^{(k)} a^k.
  \end{displaymath}
  In view of Lemma~\ref{lem:2}, the processes $\beta$ and $\beta_n$
  belong to $\Hbmo$ and
  \begin{displaymath}
    \Hnorm{\bmo}{\beta-\beta_n} \leq \sum_{k=n+1}^\infty \Hnorm{\bmo}{\zeta^{(k)}}
    \rho^k \to 0, 
    \quad n\to \infty.
  \end{displaymath}
  The bilinearity of $\widetilde F$ and Lemma~\ref{lem:1} then yield
  \begin{align*}
    &\Hnorm{\bmo}{F(\beta)-F(\beta_n)} = \Hnorm{\bmo}{\widetilde
      F(\beta,\beta)-\widetilde
      F(\beta_n,\beta_n)} \\
    & \qquad =
    \Hnorm{\bmo}{\widetilde F(\beta-\beta_n,\beta+\beta_n)}  \\
    &\qquad \leq 2\kappa \Theta
    \Hnorm{\bmo}{\beta-\beta_n}(\Hnorm{\bmo}{\beta}+\Hnorm{\bmo}{\beta_n})
    \to 0, \quad n\to \infty,
  \end{align*}
  and to conclude the proof of~\eqref{eq:17} we only have to show that
  \begin{equation}
    \label{eq:18}
    \Hnorm{\bmo}{F(\beta_n)-\beta_n} \to 0,
    \quad n\to \infty.
  \end{equation}

  From the bilinearity of $\widetilde{F}$ and the construction of
  $(\zeta^{(k)})$ we deduce that
  \begin{align*}
    F(\beta_n) - \beta_n &= \widetilde{F}(\sum_{k=1}^n \zeta^{(k)}
    a^k,\sum_{k=1}^n
    \zeta^{(k)} a^k) - \sum_{k=2}^n \zeta^{(k)} a^k \\
    &= \sum_{l,m=1}^n \widetilde{F}(\zeta^{(l)},\zeta^{(m)})a^{l+m} -
    \sum_{k=2}^n \left(\sum_{l+m=k}
      \widetilde{F}(\zeta^{(l)},\zeta^{(m)})\right)
    a^k \\
    &= \sum_{\substack{1\leq l,m\leq n \\ l+m>n}}
    \widetilde{F}(\zeta^{(l)},\zeta^{(m)})a^{l+m}.
  \end{align*}
  Using the boundedness of the bilinear map $\widetilde F$ by $2
  \kappa \Theta$ we obtain
  \begin{align*}
    &\Hnorm{\bmo}{F(\beta_n) - \beta_n} \leq 2\kappa \Theta
    \sum_{l+m>n}
    \Hnorm{\bmo}{\zeta^{(l)}}\Hnorm{\bmo}{\zeta^{(m)}} \rho^{l+m} \\
    &\qquad\leq 2\kappa \Theta \left(\left(\sum_{k=1}^\infty
        \Hnorm{\bmo}{\zeta^{(k)}} \rho^k\right)^2 -
      \left(\sum_{k=1}^{[n/2]}\Hnorm{\bmo}{\zeta^{(k)}}
        \rho^k\right)^2\right)
  \end{align*}
  and~\eqref{eq:18} follows from Lemma~\ref{lem:2}.

  The power series representation for $Y(a)$ and
  formulas~\eqref{eq:12} and~\eqref{eq:15} for its coefficients
  readily follow from the expansion for $\zeta(a)$ as soon as we write
  $Y(a)$ as
  \begin{displaymath}
    Y_t(a) = a\mathbb{E}_t[\Xi] + \mathbb{E}_t[\int_t^T \widetilde
    f(s,\zeta_s(a),\zeta_s(a))ds], \quad t\in [0,T], 
  \end{displaymath}
  and use the bilinearity of $\widetilde f$ and Lemma~\ref{lem:3}.
\end{proof}

\section{Applications to a price impact model}
\label{sec:model}

We consider the financial model of price impact studied in
\citet*{GarlPederPotes:09}, \citet{Germ:11},
and~\citet{KramkovPulido:16}. There is a representative dealer whose
preferences regarding terminal wealth are modeled by the exponential
utility
\begin{displaymath}
  U(x) = -\frac1a e^{-ax}, \quad x\in \mathbf{R}.
\end{displaymath}
The risk aversion coefficient $a>0$ defines the strength of the price
impact effect. In particular, as $a\downarrow 0$ we are getting the
classical impact-free model of mathematical finance; see
section~\ref{sec:asympt-expans-small}.

The financial market consists of a bank account and $n$ stocks. The
bank account pays zero interest rate . The stocks pay dividends $\Psi
= (\Psi^i)_{i=1,\dots,n}$ at maturity $T$; each $\Psi^i$ is a random
variable. While the terminal stocks' prices $S_T$ are always given by
$\Psi$, their intermediate values $S_t$ on $[0,T)$ are affected by an
exogenous demand process $\gamma$ through the following equilibrium
mechanism.

\begin{Definition}
  \label{def:1}
  A predictable process $\gamma$ with values in $\mathbf{R}^n$ is
  called a \emph{demand}. The demand $\gamma$ is \emph{viable} if
  there is an $n$-dimensional semimartingale of \emph{stock prices}
  $S$ with terminal value $S_T=\Psi$ such that the \emph{pricing}
  probability measure $\mathbb{Q}$ is well-defined by
  \begin{displaymath}
    \frac{d\mathbb{Q}}{d\mathbb{P}} \set \frac{U'(\int_0^T \gamma
      dS)}{\mathbb{E}[U'(\int_0^T \gamma
      dS)]} =  \frac{e^{-a\int_0^T \gamma
        dS}}{\mathbb{E}[e^{-a\int_0^T \gamma dS}]}
  \end{displaymath}
  and $S$ and the stochastic integral $\gamma\cdot S$ are uniformly
  integrable martingales under $\mathbb{Q}$.
\end{Definition}

Lemma 2.2 in~\cite{KramkovPulido:16} clarifies the economic meaning of
Definition~\ref{def:1}. It shows that a demand $\gamma$ is viable if
and only if it defines the optimal number of stocks for the dealer
trading at stock prices $S=S(\gamma)$.

Under~\ref{item:1}, for a viable demand $\gamma$ accompanied by
stocks' prices $S$ and the pricing measure $\mathbb{Q}$ there are
unique processes $\alpha\in \Hspace{0}(\mathbf{R}^d)$ and $\sigma\in
\Hspace{0}(\mathbf{R}^{n\times d})$, called, respectively, the
\emph{market price of risk} and the \emph{volatility}, such that
\begin{align*}
  \frac{d\mathbb{Q}}{d\mathbb{P}} &= e^{-\int_0^T \alpha\, dB
    -\frac12 \int_0^T \abs{\alpha_t}^2 dt}, \\
  S_t &= S_0 + \int_0^t \sigma_s \alpha_s ds + \int_0^t \sigma dB.
\end{align*}

Theorem 3.1 in~\cite{KramkovPulido:16} characterizes $S$, $\alpha$,
and $\sigma$ in terms of solutions to a system of quadratic BSDEs.
More precisely, it states that a demand $\gamma$ is viable and is
accompanied by the stock prices $S$ if and only if there are a
one-dimensional semimartingale $R$ and predictable processes $\eta\in
\Hspace{0}(\mathbf{R}^d)$, and $\theta \in
\Hspace{0}(\mathbf{R}^{n\times d})$, such that, for every $ t\in
[0,T]$,
\begin{align}
  \label{eq:19}
  aR_t &= \frac12 \int_t^T (\abs{\theta^*_s\gamma_s}^2 -
  \abs{\eta_s}^2)
  ds - \int_t^T \eta dB, \\
  \label{eq:20}
  aS_t &= a\Psi - \int_t^T \theta_s(\eta_s+\theta_s^*\gamma_s) ds -
  \int_t^T \theta dB,
\end{align}
and such that the stochastic exponential $Z \set \mathcal{E}(-(\eta +
\theta^*\gamma)\cdot B)$ and the processes $ZS$ and $Z(\gamma\cdot S)$
are uniformly integrable martingales.

In this case, $Z$ is the density process of the pricing measure
$\mathbb{Q}$, and the market price of risk $\alpha$ and the volatility
$\sigma$ are given by
\begin{align}
  \label{eq:21}
  \alpha &= \eta + \theta^*\gamma, \\
  \label{eq:22}
  \sigma &= \theta/a.
\end{align}
The value of the auxiliary process $R$ at time $t$ can be written as
\begin{equation*}
  R_t \set U^{-1}\left(\mathbb{E}_t\left[U\left(\int_t^T
        \gamma dS\right)\right]\right) \\
  = - \frac1a \log\left(\mathbb{E}_t\left[e^{-a \int_t^T \gamma
        dS}\right]\right),
\end{equation*}
and thus represents the dealer's \emph{certainty equivalent value} of
the \emph{remaining gain} $\int_t^T \gamma dS$.

\begin{Remark}
  \label{rem:1}
  From Definition~\ref{def:1} we deduce that the dependence of stocks'
  prices $S = S(\gamma,a,\Psi)$ on the viable demand $\gamma$, on the
  risk-aversion coefficient $a$, and on the dividend $\Psi$ has the
  following homogeneity properties: for $b>0$,
  \begin{equation*}
    S(b\gamma,a,\Psi) = S(\gamma,ba,\Psi) = \frac1b S(\gamma,a,b\Psi).
  \end{equation*}
  This yields similar properties of the market prices of risk $\alpha
  = \alpha(\gamma,a,\Psi)$ and of the volatilities $\sigma =
  \sigma(\gamma,a,\Psi)$:
  \begin{equation}
    \label{eq:23}
    \begin{split}
      \alpha(b\gamma,a,\Psi) &= \alpha(\gamma,ba,\Psi) =
      \alpha(\gamma,a,b\Psi), \\
      \sigma(b\gamma,a,\Psi) &= \sigma(\gamma,ba,\Psi) = \frac1b
      \sigma(\gamma,a,b\Psi).
    \end{split}
  \end{equation}
\end{Remark}

\subsection{Stability with respect to demand}
\label{sec:stability}

A demand $\gamma$ is called \emph{simple} if it has the form:
\begin{displaymath}
  \gamma=\sum_{i=0}^{m-1} \theta_{i}1_{(\tau_i,\tau_{i+1}]},
\end{displaymath}
where $0=\tau_0 < \tau_1 < \dots < \tau_m = T$ are stopping times and
$\theta_i$ is a $\mathcal{F}_{\tau_i}$-measurable random variable with
values in $\mathbf{R}^n$, $i=0,\dots,m-1$.  Theorem~1 in
\cite{Germ:11} shows that every bounded simple demand $\gamma$ is
viable provided that the dividends $\Psi = (\Psi^i)$ have all
exponential moments.  Moreover, in this case, the price process
$S=S(\gamma)$ is unique and is constructed explicitly by backward
induction.

For general (nonsimple) demands the situation is more
involved. As~\cite[Proposition 4.3]{KramkovPulido:16} shows, even for
bounded dividends $\Psi$ and demands $\gamma$, either existence or
uniqueness of prices $S=S(\gamma)$ may fail.  On a positive side, by
\cite[Theorem 4.1]{KramkovPulido:16}, there is a constant $c=c(n)>0$
(dependent only on the number of stocks $n$) such that if
\begin{equation}
  \label{eq:24}
  a \Hnorm{\infty}{\gamma} \Lnorm{\bmo}{\Psi-\mathbb{E}[\Psi]} \leq c, 
\end{equation}
then the prices $S=S(\gamma)$ exist and are unique.

The following theorem shows that under~\eqref{eq:24} the prices
$S=S(\gamma)$ are stable under small changes in the demand
$\gamma$. In particular, they can be well-approximated by the prices
originated from simple demands.

\begin{Theorem}
  \label{th:3}
  Assume \ref{item:1} and let $p > 1$. There is a constant
  $c=c(n,p)>0$ such that if $(\gamma^m)_{m\geq 1}$ and $\gamma$ are
  elements of $\Hinfty(\mathbf{R}^{n})$ such that
  \begin{equation}
    \label{eq:25}
    a \Hnorm{\infty}{\gamma^m} \Lnorm{\bmo}{\Psi-\mathbb{E}[\Psi]}
    \leq c, \quad m\geq 1,
  \end{equation}
  and
  \begin{displaymath}
    \mathbb{E}\left[\int_0^T \abs{\gamma^m_t - \gamma_t} dt\right]
    \to 0, \quad m\to \infty,
  \end{displaymath}
  then $(\gamma^m)_{m\geq 1}$ and $\gamma$ are viable demands and the
  corresponding stock prices $(S^m)_{m\geq 1}$ and $S$, volatilities
  $(\sigma^m)_{m\geq 1}$ and $\sigma$, and the market prices of risk
  $(\alpha^m)_{m\geq 1}$ and $\alpha$ converge as
  \begin{equation}
    \label{eq:26}
    \Snorm{p}{S^m - S} + \Hnorm{p}{\sigma^m - \sigma} + \Hnorm{p}{\alpha^m -
      \alpha} \to 0, \quad m\to \infty.
  \end{equation}
\end{Theorem}
\begin{proof}
  Observe that the self-similarity relations~\eqref{eq:23} for the
  market prices of risk and volatilities allow us to assume that
  \begin{displaymath}
    a = 1 \geq \Hnorm{\infty}{\gamma^m}, \quad m\geq 1.
  \end{displaymath}
  Clearly, $\gamma$ satisfies~\eqref{eq:24} with the same constant $c$
  as in~\eqref{eq:25}. By~\cite[Theorem 4.1]{KramkovPulido:16}, we can
  choose $c=c(n)$ so that the demands $(\gamma^m)$ and $\gamma$ are
  viable and are accompanied by unique stock prices.  Using
  Theorem~\ref{th:1} and the BSDE
  characterizations~\eqref{eq:19}--\eqref{eq:22} we can also choose
  $c=c(n,p)$ so that
  \begin{displaymath}
    \Snorm{p}{S^m - S} + \Hnorm{p}{\sigma^m - \sigma} + \Hnorm{p}{\alpha^m -
      \alpha} \leq C
    \Lnorm{p}{\int_0^T \abs{\gamma^m_t - \gamma_t} (\abs{\alpha_t}^2
      +\abs{\sigma_t}^2) dt} 
  \end{displaymath}
  for some $C=C(n,p)$. This yields~\eqref{eq:26} by the dominated
  convergence theorem.
\end{proof}

\begin{Remark}
  Using Theorem~\ref{th:1} we can prove an analogous convergence
  result of prices $(S^m)$ and their local characteristics $(\alpha^m,
  \sigma^m)$ in $\bmo$ provided that
  $\Hnorm{\infty}{\gamma^m-\gamma}\to 0$. It is important to
  highlight, however, that simple demands are not dense in the space
  of bounded demands with respect to the norm $\Hnorm{\infty}{\cdot}$.
\end{Remark}

\subsection{Asymptotic expansion for small risk-aversion}
\label{sec:asympt-expans-small}

As the risk aversion coefficient $a$ approaches zero, the price impact
effect vanishes and we obtain a classical model of mathematical
finance.  Theorem~\ref{th:4} below provides analytic expansions of
volatilities and market prices of risk in the neighborhood of
$a=0$. The terms in these expansions are computed recursively, by
martingale representation and, thus, are quite explicit.

We write $z\in \mathbf{R}^{(n+1)\times m}$ as $z=(z_1,z_2)$ with
$z_1\in \mathbf{R}^m$ and $z_2 \in \mathbf{R}^{n\times m}$, the
decomposition of $(n+1)\times m$-dimensional matrix on its first and
subsequent rows; hereafter $m=1$ or $d$, the dimension of underlying
Brownian motion from~\ref{item:1}.

For a vector $w\in \mathbf{R}^n$ consider the bilinear map:
\begin{displaymath}
  \map{g(\cdot;w) = (g_1,g_2)(\cdot;w)}{\mathbf{R}^{(n+1)\times d}\times
    \mathbf{R}^{(n+1)\times d}}{\mathbf{R}^{n+1}}
\end{displaymath}
defined for $u=(u_1,u_2)$ and $v=(v_1,v_2)$ from
$\mathbf{R}^{(n+1)\times d}$ as
\begin{align*}
  g_1(u,v;w) &\set \frac12(\ip{u_2^*w}{v_2^*w}
  - \ip{u_1}{v_1}) \\
  g_2(u,v;w) & \set -\frac12(u_2v_1 + v_2u_1 + (u_2v_2^* +
  v_2u_2^*)w),
\end{align*}
where $\ip{x}{y}$ denotes the scalar product of $x,y\in \mathbf{R}^m$.

Take $\gamma\in \Hinfty(\mathbf{R}^n)$. Lemma~\ref{lem:1} shows that
for $\mu$ and $\nu$ in $\Hbmo(\mathbf{R}^{(n+1)\times d})$ there is a
unique $\zeta \in \Hbmo(\mathbf{R}^{(n+1)\times d})$ such that
\begin{equation}
  \label{eq:27}
  \int_0^t \zeta dB = \mathbb{E}_t[\int_0^T g(\mu_s,\nu_s;\gamma_s)ds] -
  \mathbb{E}[\int_0^T g(\mu_s,\nu_s;\gamma_s)ds].
\end{equation}
Hence, we can define a bilinear map
\begin{displaymath}
  \map{G(\cdot,\cdot;\gamma)}{\Hbmo(\mathbf{R}^{(n+1)\times d})\times
    \Hbmo(\mathbf{R}^{(n+1)\times d})}{\Hbmo(\mathbf{R}^{(n+1)\times
      d})}
\end{displaymath}
with $\zeta = G(\mu,\nu;\gamma)$ given by \eqref{eq:27}.

Denote also by $S(0)$ and $\sigma(0)$ the unperturbed stocks' prices
and volatilities corresponding to the case $\gamma=0$:
\begin{displaymath}
  S_t(0)= \mathbb{E}_t[\Psi] = S_0(0) + \int_0^t \sigma(0) dB, \quad
  t\in [0,T].
\end{displaymath}

\begin{Theorem}
  \label{th:4}
  Assume~\ref{item:1} and that
  \begin{displaymath}
    0<\Lnorm{\bmo}{\Psi-\mathbb{E}[\Psi]}<\infty.
  \end{displaymath}
  There is a constant $c=c(n)>0$ such that if $\gamma \in
  \Hinfty(\mathbf{R}^n)$, $\gamma\not=0$, and the risk-aversion
  satisfies
  \begin{equation}
    \label{eq:28}
    0<a < \rho \set
    \frac{c}{\Hnorm{\infty}{\gamma} \Lnorm{\bmo}{\Psi - \mathbb{E}[\Psi]}},
  \end{equation}
  then $\gamma$ is a viable demand. The price $S(\gamma;a)$ is unique
  and admits the power series expansion in $\Sbmo$:
  \begin{displaymath}
    S(\gamma;a) = S(0) + \sum_{k=1}^\infty S^{(k)} a^k, \quad
    a<\rho. 
  \end{displaymath}
  The market price of risk $\alpha(\gamma;a)$ and the volatility
  $\sigma(\gamma;a)$ have the power series expansions in $\Hbmo$:
  \begin{align*}
    \alpha(\gamma;a) &= \sum_{k=1}^\infty (\zeta_1^{(k)} +
    (\zeta_2^{(k)})^*\gamma) a^k,  \quad a<\rho,\\
    \sigma(\gamma;a) &= \sum_{k=0}^\infty \zeta_2^{(k+1)} a^k, \quad
    a<\rho.
  \end{align*}
  Here the coefficients $\zeta^{(k)} = (\zeta_1^{(k)},\zeta_2^{(k)})$,
  $k\geq 1$, in $\Hbmo(\mathbf{R}^{(n+1)\times d})$ are given
  recursively as
  \begin{align}
    \zeta^{(1)} &= (\zeta^{(1)}_1, \zeta^{(1)}_2) =
    (0,\sigma(0)), \label{eq:29}\\
    \zeta^{(k)} &= \sum_{l+m=k} {G}(\zeta^{(l)},\zeta^{(m)};\gamma),
    \quad k\geq 2, \label{eq:30}
  \end{align}
  and the coefficients $(S^{(k)})_{k\geq 1}
  \subset\Sbmo(\mathbf{R}^n)$ are given by
  \begin{displaymath}
    S^{(k)}_t = -\sum_{l+m=k+1} \mathbb{E}_t[\int_t^T \zeta_{2,s}^{(l)}
    (\zeta_{1,s}^{(m)} + (\zeta_{2,s}^{(m)})^*\gamma_s) ds],
    \quad t\in [0,T].
  \end{displaymath}
\end{Theorem}

\begin{Remark}
  \label{rem:2}
  The leading price impact coefficient in the expansion for stock
  prices is given by
  \begin{displaymath}
    S^{(1)}_t = -\mathbb{E}_t\left[\int_t^T\sigma_s(0)
      \sigma_s(0)^*\gamma_s\,ds\right], \quad t\in [0,T]. 
  \end{displaymath}
  This result had been obtained earlier in~\cite[Theorem~2]{Germ:11}
  for a simple demand.
\end{Remark}

\begin{proof}[Proof of Theorem~\ref{th:4}.]
  Theorem 4.1 in~\cite{KramkovPulido:16} yields the existence of a
  constant $c=c(n)$ such that, under~\eqref{eq:28}, $\gamma$ is a
  viable demand accompanied by the unique price process $S = S(\gamma,
  a)$. Observe that the dependence of the coefficients $\zeta^{(k)}$
  on $\gamma$ and $\sigma(0)$ has the homogeneity properties:
  \begin{align*}
    \zeta_1^{(k)}(b\gamma,\sigma(0)) &=
    \zeta_1^{(k)}(\gamma,b\sigma(0)) =
    b^k\zeta_1^{(k)}(\gamma,\sigma(0)), \\
    b\zeta_2^{(k)}(b\gamma,\sigma(0)) &=
    \zeta_2^{(k)}(\gamma,b\sigma(0)) =
    b^k\zeta_2^{(k)}(\gamma,\sigma(0)), \quad b>0,
  \end{align*}
  which can be verified by induction. In view of these properties and
  the self-similarity relations~\eqref{eq:23} for the market prices of
  risk and volatilities, we can assume, without a loss in generality,
  that
  \begin{displaymath}
    \Hnorm{\infty}{\gamma} = \Hnorm{\bmo}{\sigma(0)} = 1.
  \end{displaymath}
  In this case, the stochastic bilinear map $g(\cdot,\cdot;\gamma)$ is
  bounded by a constant $\Theta = \Theta(n)$ such that
  \begin{displaymath}
    \abs{g(u,v;w)} \leq \Theta \abs{u}\abs{v} \;\text{for every}\;
    w\in \mathbf{R}^n\;\text{with}\;\abs{w}\leq 1.
  \end{displaymath}

  Taking now the constant $c$ also smaller than $1/(8\Theta \kappa)$,
  where $\kappa = \kappa(n)$ is defined in~\eqref{eq:2}, we deduce
  from Theorem~\ref{th:2} the existence and uniqueness of
  $\eta=\eta(a)\in \Hbmo(\mathbf{R}^d)$ and $\theta=\theta(a)\in
  \Hbmo(\mathbf{R}^{n\times d})$ solving~\eqref{eq:19}--\eqref{eq:20}
  and such that
  \begin{displaymath}
    \sqrt{\Hnorm{\bmo}{\eta}^2 + \Hnorm{\bmo}{\theta}^2} \leq
    \frac1{4\kappa\Theta}.
  \end{displaymath}
  Moreover, the maps $a\mapsto \eta(a)$ and $a\mapsto \theta(a)$ of
  $(-\rho,\rho)$ to $\Hbmo$ are analytic, and their power series are
  given by
  \begin{align*}
    \eta(a) &= \sum_{k=1}^\infty \zeta_1^{(k)} a^k, \\
    \theta(a) &= \sum_{k=1}^\infty \zeta_2^{(k)} a^k,
  \end{align*}
  with the coefficients $\zeta^{(k)} = (\zeta_1^{(k)},\zeta_2^{(k)})$,
  $k\geq 1$, in $\Hbmo(\mathbf{R}^{(n+1)\times d})$ determined by
  \eqref{eq:29}--\eqref{eq:30}. The power series expansion for $S$
  follow from Theorem~\ref{th:2}, while the expansions for $\alpha$
  and $\sigma$ follow from the linear invertibility
  relations~\eqref{eq:21}--\eqref{eq:22} between $(\alpha,\sigma)$ and
  $(\eta,\theta)$.
\end{proof}

\bibliographystyle{plainnat}

\bibliography{../bib/finance}

\end{document}